\documentclass{article}

\usepackage[letterpaper,top=2cm,bottom=2cm,left=3cm,right=3cm,marginparwidth=1.75cm]{geometry}

\usepackage{amsmath,amsthm,amssymb,enumitem,todonotes,cancel}

\newtheorem{theorem}{Theorem} 
\newtheorem{lemma}{Lemma}

\theoremstyle{remark}
\newtheorem*{remark}{Remark}

\newcommand{\nat}{\mathbb{N}}

\newcommand{\zo}{\{0,1\}}
\newcommand{\mapping}{\rightarrow}

\newcommand{\samp}{\textsf{Samp}}
\newcommand{\compress}{\mathsf{Compress}}
\newcommand{\inv}{\mathsf{Inverter}}

\newcommand{\aeowf}{$\exists$ OWF\,}
\newcommand{\noaeowf}{\cancel{$\exists$} OWF\,}

\usepackage{graphicx}
\usepackage[colorlinks=true, allcolors=blue]{hyperref}
\usepackage{cleveref}

\DeclareMathOperator{\prob}{{Prob}}
\title{On one-way functions  and the average time complexity of almost-optimal compression}
\author{Marius Zimand\thanks{Department of Computer and Information Sciences, Towson University, Baltimore, MD. Partially supported by a grant from the School of Emerging Technologies at Towson University.}}

\date{}
\begin{document}
\maketitle

\begin{abstract}
We show that one-way functions exist if and only if there exists an efficiently-samplable distribution relative to which almost-optimal compression is hard on average.  The result is obtained by combining a theorem of  Ilango, Ren, and Santhanam~\cite{ila-ren-san:t:owf-kolm, ila-ren-san:c:owf-kolm} and one by  Bauwens and Zimand~\cite{bau-zim:j:univcompression}.
\end{abstract}

\section{Introduction}

Several recent papers show that the existence of  one-way functions (OWF) is equivalent to the hardness of certain problems in meta-complexity~\cite{liu-pass:c:OWF-Kolm,liu-pass:c:OWF-2021,ren-san:t:KT-crypto, ila-ren-san:t:owf-kolm,ila-ren-san:c:owf-kolm,liu-pass:t:OWF-timeKolm,liu-pass:c:probKolm,hir-lu-oli:c:pKt,lu-san:c:impagliazzo}. 
The motivation for this research line comes primarily from cryptography, where  one-way functions play a central role\footnote{See~\cite{hir-lu-oli:c:pKt} and~\cite{lu-san:c:impagliazzo} for a discussion of some of these and related works.}. 
Ilango, Ren and Santhanam~\cite{ila-ren-san:t:owf-kolm,ila-ren-san:c:owf-kolm} have obtained a result of this type involving standard (unbounded) Kolmogorov complexity. Informally speaking,  they have shown that one-way functions exist if and only if ``finding good approximations of Kolmogorov complexity" is hard on average with respect to some polynomial-time samplable distribution.  Bauwens and Zimand~\cite{bau-zim:j:univcompression} have shown that given a good approximation of the Kolmogorov complexity of a string $x$, one can compress $x$ in probabilistic polynomial time to a string of length close to its complexity (so, $x$ is almost-optimally compressed). The combination of these 2 results yields the following theorem. 
\begin{theorem}[Informal statement]
\label{t:main}
  The following two assertions are equivalent:
  \begin{enumerate}
      \item There exists a one-way function. 
      \item Almost optimal compression is hard on average with respect to some polynomial-time samplable distribution.

  \end{enumerate}
\end{theorem}
The result of Ilango et.\,al.\,~\cite{ila-ren-san:t:owf-kolm,ila-ren-san:c:owf-kolm} is not exactly stated in the form that we mentioned above. For this reason, we prefer to give a proof which does not directly invoke~\cite{ila-ren-san:t:owf-kolm,ila-ren-san:c:owf-kolm}, but which closely follows their method. In one direction, it is based on  results of Impagliazzo, Levin and Luby\cite{imp-lev:c:hardinstances,imp-lub:c:owf-and-crypto} connecting the existence of OWFs to the hardness of approximating poly-time samplable distributions, and, in the other direction, it is based on the connection between OWFs and pseudo-random generators established by H{\aa}stad, Impagliazzo, Levin and Luby~\cite{ha-im-lev-lub:j:prgenonewayf}.

\section{Definitions, and technical tools}
\textbf{\emph{Kolmogorov complexity.}}
We fix an optimal universal Turing machine $U$ with prefix-free domain. A program  for string $x$ is a string $p$ such that $U(p)=x$. The prefix-free Kolmogorov complexity $K(x)$ of the string $x$ is the length of a shortest program for $x$.\footnote{The prefix-free Kolmogorov complexity $K(x)$ is a little more convenient for the proof than the plain complexity $C(x)$. The difference $K(x) - C(x)$ is bounded by $2 \log |x|$ and, therefore, the result is valid for $C(x)$ as well.}
\medskip

\noindent
\textbf{\emph{Distributions.}} We consider ensembles of distributions. An ensemble has the form $D = (D_n)_{n \in \nat}$, where each $D_n$ is a distribution on $\zo^n$. The ensemble $D$ is \emph{samplable} if there exists a probabilistic  algorithm $\samp$, such that for every $n$ and every $x \in \zo^n$, 
\[
\prob ~[\samp(1^n) = x] = D_n(x)
\]
(the probability is over the randomness of $\samp$).
\smallskip

\noindent
$D$ is said to be \emph{P-samplable}, in case  $\samp$ runs in polynomial time.
\smallskip

\noindent
Some notation: For every $x$, we denote $D(x) = D_{|x|}(x)$.
For every $m$, $U_m$ denotes the uniform distribution over $m$-bit strings.
\begin{lemma}
\label{l:coding-upperbound} If $D$ is samplable, then for every $x$ in its support,
\[
K(x) \leq \log \frac{1}{D(x)} + 3 \log (|x|) + O(1).
\]
\end{lemma}
\begin{proof}
    Fix a binary string $x$ and let $n$ be its length. Given $n$ and the code of $\samp$, one can compute $D_n(y)$ for all strings $y$ of length $n$ and then list all these strings in descending order of their $D_n(\cdot)$ probability (with ties broken, say, lexicographically).   The string $x$ is described by its rank $t$ in this list. Since the $D_n$-probability of the first $t$ strings in the order is at most $1$ and at least $t \cdot D_n(x)$, it follows that $t\leq \lceil 1/D_n(x) \rceil$. An overhead of $2 \log (|x|) + O(1)$ bits is added to obtain a self-delimited description in the standard way.
\end{proof}
\begin{lemma}
\label{l:coding-lowerbound} For every distribution $D$, and every $\Delta \geq 0$,
\[
\prob_{x \leftarrow D} ~[K(x) \geq \log \frac{1}{D(x)} - \Delta] \geq 1-2^{-\Delta}.
\]
\end{lemma}
\begin{proof}
The complement of the event in the probability is $E= \{x \mid D(x) \leq 2^{-\Delta} \cdot 2^{-K(x)}\}$.
We have
\[
D(E) = \sum_{x \in E} D(x) \leq \sum_{x \in E} 2^{-\Delta} \cdot 2^{-K(x)} \leq 2^{-\Delta} \sum_{x \in \zo^*}
2^{-K(x)} \leq 2^{-\Delta} \cdot 1 = 2^{-\Delta}.\]
In the penultimate transition, we have used the Kraft inequality, which is legitimate because $K( \cdot)$ represents the lengths of a prefix-free code.
\end{proof}
\medskip

\noindent
\textbf{\emph{Formal statement of~\Cref{t:main}. }}
The following 2 assertions are equivalent: 
\smallskip

\noindent
\textbf{\emph{Assertion (1): The hypothesis ``\,\aeowf":}} There exists a polynomial-time computable
$f: \zo^* \mapping \zo^*$ with the following property:  For every probabilistic polynomial-time algorithm $\inv$, every $q \in \nat$ and almost every length $n \in \nat$,
\[
\prob_{x \leftarrow U_n, \inv} [\inv(1^n, f(x))  \in f^{-1}(f(x))] \leq 1/n^q.
\]
(The notation $\prob_{x \leftarrow U_n, \inv}$ means that the probability is over $U_n \times \textrm{randomness of } \inv$.)



\smallskip

\noindent
\textbf{\emph{Assertion (2): The hypothesis ``almost optimal compression is hard on average" :}}
 There exists a $P$-samplable distribution $D$ and a constant $c$ with the following property: For every probabilistic polynomial-time algorithm $\compress$, at almost every length $n$, 
      \[
\prob_{x \leftarrow D_n, \compress} [\compress\,(x) \text{ outputs a program of $x$ of length } \geq K(x) + c\log^2 n] > 1/100.
      \]
(The event in the probability expresses the failure of almost optimal compression, and thus assertion (2) states that for any efficient algorithm this failure happens with significant probability.)
\bigskip

      \begin{remark}
The ``infinitely often" version of~\Cref{t:main} is also true, with essentially the same proof. More precisely, if we modify Assertions 1 and 2 by replacing ``almost every length $n$" with ``infinitely many lengths $n$," the modified assertions are also equivalent.
\smallskip

Also, the version in which the additive overhead $c \log^2 n$ is replaced by $n^\gamma$ (for every $\gamma \in (0,1)$ ) is true with essentially the same proof.
\end{remark}
\bigskip

\noindent
\textbf{\emph{Results from the literature that we use.}}
\begin{theorem}[\cite{imp-lub:c:owf-and-crypto,  imp-lev:c:hardinstances}; this variant is stated and proved in~\cite{ila-ren-san:t:owf-kolm}]
\label{t:impag-levin} Assume  the hypothesis ``\aeowf" is not true.
Let $D = (D_n)_{n \in \nat}$ be a P-samplable ensemble of distributions, and $q \in \nat$. There exists a probabilistic polynomial-time algorithm $A$ and a constant $c > 1$  such that for infinitely many $n$,
\[
\prob_{x \leftarrow D_n, A} ~[D_n(x)/c \leq A(x) \leq D_n(x)] \geq 1 - \frac{1}{n^q}.
\]
    \end{theorem}
\noindent
In other words: If there are no one-way functions, then P-samplable distributions can be approximated efficiently in the average sense.
\begin{theorem}[\cite{bau-zim:j:univcompression}] 
\label{t:bau-zim}
There exists a probabilistic polynomial-time algorithm $\compress$ that for every input triple $(x \in \zo^*, m \in \nat, \textrm{rational } \epsilon > 0)$ outputs with probability 1 a string $z$ of length $m + O(\log m \cdot \log |x|/\epsilon)$ and if $m \geq K(x)$ then 
\[
\prob_{\compress} [\textrm{$z$ is a program for $x$}] \geq 1-\epsilon. 
\]
\end{theorem}
\noindent
In other words: Given a good approximation of the Kolmogorov complexity of a string $x$, one can efficiently compress $x$ almost optimally (where efficiently means in probabilistic polynomial time).
 \section{Proof of~\Cref{t:main}}
 
\textbf{Proof of assertion (2) $\rightarrow$ assertion (1).} 

We actually prove the contrapositive: \noaeowf $\Rightarrow$ $\neg$ assertion (2) (i.e, almost optimal compression is easy on average).

Let $D = (D_n)_{n \in \nat}$ be a P-samplable ensemble. By~\Cref{l:coding-lowerbound} and~\Cref{l:coding-upperbound}, for some constant $c$, for every $n$
\[
\prob_{x \leftarrow D_n}[ \log \frac{1}{D_n(x)} - c \log n \leq K(x) \leq \log \frac{1}{D_n(x)} + c \log n] \geq 1 - 1/n.
\]
Under our assumption ``\noaeowf," \Cref{t:impag-levin} states that there exists an  algorithm that approximates $D_n(x)$ with high probability, and therefore it also approximates $K(x)$ with high probability. More precisely, by rescaling, we get a probabilistic polynomial-time algorithm $A$ such that, for every $n$,
\[
\prob_{x \leftarrow D_n, A}[ K(x) \leq A(x) \leq K(x) + c \log n] \geq 1 - 1/n.
\]
Then, the algorithm $\compress$ from~\Cref{t:bau-zim} with $m = A(x)$ and $\epsilon = 1/200$ shows the invalidity of assertion (2).
\medskip

\noindent
\textbf{Proof of assertion (1) $\rightarrow$ assertion (2).} 

(\aeowf $\Rightarrow$  almost optimal compression is hard on average.)
\medskip

The idea is that an efficient  good compressor would break the security of any candidate pseudo-random generator (p.r.g.), because the output of the generator can be compressed to a much shorter string, whereas a genuinely random string cannot. Therefore, pseudorandom generators would not exist and hence there would be  no OWF, contradicting assertion (1).
Now, the details. 


Suppose ``\,\aeowf" is true. Then, by~\cite{ha-im-lev-lub:j:prgenonewayf} combined with the methods to obtain ensembles of p.r.g.'s with every possible output length~\cite[Section 3.3.3]{gol:b:crypto},  there exists an ensemble of p.r.g.'s $G = (G_n)_{n \in \nat}$, computable in polynomial time (uniformly), with $G_n : \zo^{s(n)} \mapping \zo^{n}$, where the seed length $s(n)$ is bounded by $n^{1/3}$, that satisfies the following security guarantee:  For every probabilistic polynomial-time algorithm  $T$ (the hypothetical distinguisher) and every $q \in \nat$, 
for almost every $n \in \nat$, the probabilities that (a) $T$ accepts $G_n (U_{s(n)})$ and (b) $T$ accepts $U_{n}$, differ by at most $1/n^q$.  
\smallskip

Consider the following P-samplable distribution $D_{n}$: 

\quad\quad with probability $1/2$, output $G (U_{s(n)})$ and with probability 1/2, output $U_{n}$. 
\smallskip

Clearly, if assertion (2) is false, then there exists a constant $c$, a probabilistic polynomial-time algorithm $A$ (derived from $\compress\textbf{}$ in the straightforward way), and an infinite set $B$, so  that at every length $n \in B$,  with $D_n \times \textrm{(randomness of $A$)}$-probability $\geq 1-1/100$  approximates $K(x)$ with slack at most $c\log^2 n$. 
Let
 \[
 \textrm{BAD}=\{x \in \{0,1\}^n \mid \prob_{A}[\,|A(x) - K(x)| \ge c \log^2 n\,] \ge 5/100\}.
 \]
(In other words: BAD is the event
which says that $A$ fails to approximate $K$ with significant probability over the randomness of $A$.)

Then, for every $n \in B$, by Markov's inequality,
\begin{equation*}
\label{e:strong}
\prob_{x \leftarrow D_n} [\textrm{BAD} ] \le 1/5,
 \end{equation*}
  
By inspecting the sampling procedure, we see that each element $x$ in BAD has $D_{n}$-probability mass at least $(1/2) \cdot 2^{-n}$  
and thus $1/5 \geq D_{n}(\textrm{BAD}) \geq  (\# BAD) \cdot (1/2 \cdot 2^{-n}) = 1/2 \cdot \prob_{U_{n}} [BAD]$, and so 
\[
\prob_{U_{n}} [BAD]\leq 2/5.
\] 
Also, each element in $\textrm{BAD} \,\cap\, \textrm{Im}(G(U_{s(n)})$ has $D_{n}$-probability mass at least $(1/2) \cdot 2^{-s(n)}$, which, similarly to the above,  implies that  
\[
\prob_{U_{s(n)}} [BAD \,\cap\, \textrm{Im}(G(U_{s(n)}))] \leq 2/5.
\]

We now define the probabilistic polynomial-time distinguisher $T$: $T$ on input $z$ of length $n$, executes $A$ on input $z$, and \emph{accepts}, if $A(z) \le  n^{1/2}$, and \emph{rejects} otherwise. Note that $G(U_{s(n)})$ with probability
1 has prefix-free complexity at most $s(n) + 2 \log s(n) + O(1) \leq n^{1/3}+ O(\log n)$. Therefore, if $G(U_{s(n)}) \not \in \textrm{BAD}$ and if $A$ uses randomness that yields good approximation, then $T$ accepts $G(U_{s(n)})$. Also, $U_{n}$, with probability at least 1-1/n, has complexity at least $n- \log n$. If this is the case, and $U_n \not \in \textrm{BAD}$, and $A$ uses correct randomness, then $T$ rejects $U_n$.

Therefore, for every  $n \in B$, 
\[
\prob_{U_{s(n)},T}[T \textrm{ accepts } G(U_{s(n)})] \geq  (1 - 2/5) \cdot (1 - 5/100) = 57/100
\]
(the probability that $G(U_{s(n)})$ is not in BAD is at least $1-2/5$ and for strings not in BAD the probability that $A$ uses correct randomness is at least $1-5/100$).  

Also, for every $n \in B$,
\[
\prob_{U_{n},T}[ T \textrm{ accepts } U_{n}] \leq 1/n + 2/5 + 5/100 = 45/100 + 1/n.
\]
\noindent
($1/n$ is the probability that $U_n$ has complexity less than $n- \log n$, 2/5 is the probability that $U_n$ is in BAD, and $5/100$ is the probability conditioned on $U_n \not \in \textrm{BAD}$ that $A$ is using wrong randomness).

We are done, because the two inequalities
contradict the security of $G$. 
\medskip

\bibliographystyle{alpha}
\bibliography{theory-3}

\begin{thebibliography}{HILL99}

\bibitem[BZ23]{bau-zim:j:univcompression}
Bruno Bauwens and Marius Zimand.
\newblock Universal almost optimal compression and {S}lepian-{W}olf coding in probabilistic polynomial time.
\newblock {\em J. ACM}, 70(2):1--33, 2023.
\newblock (arxiv version posted in 2019).

\bibitem[Gol01]{gol:b:crypto}
O.~Goldreich.
\newblock {\em Foundations of Cryptography}.
\newblock Cambridge University Press, 2001.

\bibitem[HILL99]{ha-im-lev-lub:j:prgenonewayf}
Johan H{\aa}stad, Russell Impagliazzo, Leonid Levin, and Michael Luby.
\newblock A pseudorandom generator from any one-way function.
\newblock {\em SIAM Journal on Computing}, 28(4):1364--1396, 1999.
\newblock A preliminary version appeared in 21rst STOC, 1989.

\bibitem[HLO24]{hir-lu-oli:c:pKt}
Shuichi Hirahara, Zhenjian Lu, and Igor~C Oliveira.
\newblock One-way functions and pkt complexity.
\newblock {\em Cryptology ePrint Archive}, 2024.

\bibitem[IL89]{imp-lub:c:owf-and-crypto}
Russell Impagliazzo and Michael Luby.
\newblock One-way functions are essential for complexity based cryptography (extended abstract).
\newblock In {\em 30th Annual Symposium on Foundations of Computer Science, Research Triangle Park, North Carolina, USA, 30 October - 1 November 1989}, pages 230--235. {IEEE} Computer Society, 1989.

\bibitem[IL90]{imp-lev:c:hardinstances}
Russell Impagliazzo and Leonid~A. Levin.
\newblock No better ways to generate hard {NP} instances than picking uniformly at random.
\newblock In {\em 31st Annual Symposium on Foundations of Computer Science, St. Louis, Missouri, USA, October 22-24, 1990, Volume {II}}, pages 812--821. {IEEE} Computer Society, 1990.

\bibitem[IRS21]{ila-ren-san:t:owf-kolm}
Rahul Ilango, Hanlin Ren, and Rahul Santhanam.
\newblock Hardness on any samplable distribution suffices: New characterizations of one-way functions by meta-complexity.
\newblock {\em Electron. Colloquium Comput. Complex.}, {TR21-082}, 2021.

\bibitem[IRS22]{ila-ren-san:c:owf-kolm}
Rahul Ilango, Hanlin Ren, and Rahul Santhanam.
\newblock Robustness of average-case meta-complexity via pseudorandomness.
\newblock In Stefano Leonardi and Anupam Gupta, editors, {\em {STOC} '22: 54th Annual {ACM} {SIGACT} Symposium on Theory of Computing, Rome, Italy, June 20 - 24, 2022}, pages 1575--1583. {ACM}, 2022.

\bibitem[LP20]{liu-pass:c:OWF-Kolm}
Yanyi Liu and Rafael Pass.
\newblock On one-way functions and kolmogorov complexity.
\newblock In {\em 2020 IEEE 61st Annual Symposium on Foundations of Computer Science (FOCS)}, pages 1243--1254. IEEE, 2020.

\bibitem[LP21]{liu-pass:c:OWF-2021}
Yanyi Liu and Rafael Pass.
\newblock On the possibility of basing cryptography on exp $\ne$ bpp exp$\ne$ bpp.
\newblock In {\em Advances in Cryptology--CRYPTO 2021: 41st Annual International Cryptology Conference, CRYPTO 2021, Virtual Event, August 16--20, 2021, Proceedings, Part I 41}, pages 11--40. Springer, 2021.

\bibitem[LP23a]{liu-pass:t:OWF-timeKolm}
Yanyi Liu and Rafael Pass.
\newblock On one-way functions and the worst-case hardness of time-bounded kolmogorov complexity.
\newblock {\em Cryptology ePrint Archive}, 2023.

\bibitem[LP23b]{liu-pass:c:probKolm}
Yanyi Liu and Rafael Pass.
\newblock One-way functions and the hardness of (probabilistic) time-bounded kolmogorov complexity wrt samplable distributions.
\newblock In {\em Annual International Cryptology Conference}, pages 645--673. Springer, 2023.

\bibitem[LS24]{lu-san:c:impagliazzo}
Zhenjian Lu and Rahul Santhanam.
\newblock Impagliazzo’s worlds through the lens of conditional kolmogorov complexity.
\newblock In {\em 51st International Colloquium on Automata, Languages, and Programming (ICALP 2024)}, pages 110--1. Schloss Dagstuhl--Leibniz-Zentrum f{\"u}r Informatik, 2024.

\bibitem[RS21]{ren-san:t:KT-crypto}
Hanlin Ren and Rahul Santhanam.
\newblock Hardness of kt characterizes parallel cryptography.
\newblock {\em Cryptology ePrint Archive}, 2021.

\end{thebibliography}

\end{document}